\documentclass[11pt, sigconf, svgnames]{acmart}

\usepackage{booktabs} 
\usepackage{graphicx}
\usepackage{multirow}
\usepackage{color, xcolor}
\graphicspath{{./figures/}}
\usepackage{hyperref}
\usepackage{booktabs}
\usepackage{dsfont}

\setcopyright{none}

\newcommand{\Ind}{\mathds{1}}
\newcommand{\R}{\mathbb{R}}
\renewcommand{\P}{\mathbb{P}}

\acmDOI{}

\acmISBN{}

\acmConference[]{FATML Workshop}{August 2017}{Halifax, NS, Canada} 
\acmYear{2017}
\copyrightyear{}

\acmPrice{}

\renewcommand\footnotetextcopyrightpermission[1]{}

\begin{document}

  
\title{Fairer and more accurate, but for whom?}

\author{Alexandra Chouldechova}
\affiliation{%
  \institution{H. John Heinz III College \\ Carnegie Mellon University}
}
\email{achould@cmu.edu}

\author{Max G'Sell}
\affiliation{%
  \institution{Department of Statistics \\ Carnegie Mellon University}
}
\email{mgsell@cmu.edu}

\renewcommand{\shortauthors}{A. Chouldechova and M. G'Sell}

\begin{abstract}
Complex statistical machine learning models are increasingly being used or considered for use in high-stakes decision-making pipelines in domains such as financial services, health care, criminal justice and human services.  These models are often investigated as possible improvements over more classical tools such as regression models or human judgement.  While the modeling approach may be new, the practice of using some form of risk assessment to inform decisions is not.  

When determining whether a new model should be adopted, it is therefore essential to be able to compare the proposed model to the existing approach across a range of task-relevant accuracy and fairness metrics.  Looking at overall performance metrics, however, may be misleading.  Even when two models have comparable overall performance, they may nevertheless disagree in their classifications on a considerable fraction of cases. 

 In this paper we introduce a model comparison framework for automatically identifying subgroups in which the differences between models are most pronounced.  Our primary focus is on identifying subgroups where the models differ in terms of fairness-related quantities such as racial or gender disparities.  We present experimental results from a recidivism prediction task and a hypothetical lending example.  
\end{abstract}

%
%



\maketitle

\section{Introduction}

Actuarial and clinical assessments of risk have long been mainstays of decision making in domains such as criminal justice, health care and human services.  Within the criminal justice system, for instance, recidivism prediction instruments and judicial discretion commonly enter into decisions concerning bail, parole and sentencing.  In these high-stakes settings, decisions made based on erroneous predictions can have a direct adverse impact on individuals’ lives.  Institutions are therefore continually seeking to improve the accuracy of their risk predictions, and many are turning to proprietary commercial tools and more complex ``black-box'' prediction models in pursuit of accuracy gains.  

When determining whether to replace or augment an existing risk assessment method, it is important to compare the proposed model to the existing approach across a range of task-relevant accuracy and fairness metrics.  As we will demonstrate, a comparison that looks only at overall performance can present an incomplete and potentially misleading picture.  

\vspace{0.5em}

\noindent \textbf{A motivating example.}  In May 2016 an investigative journalism team at ProPublica released a report \cite{propublica2016} on a proprietary recidivism prediction instrument called COMPAS\cite{compasfaq}, developed by Northpointe Inc.  The data set \cite{propublica2016data} released as part of this report contains COMPAS decile scores, 2-year recidivism outcomes and a number of demographic and crime-related variables for defendants scored as part of pre-trial proceedings in Broward County, Florida.  In particular, the data set contains information on the number of prior offenses (hereon denoted \texttt{Priors}) for each defendant.  Since criminal history is itself a good predictor of future recidivism, it is reasonable to suppose that before \texttt{COMPAS} was introduced, judges could have based their risk assessments on \texttt{Priors} instead.  Our question is thus: \emph{Does} \texttt{COMPAS} \emph{produce more accurate (and/or equitable) predictions of recidivism than} \texttt{Priors} \emph{alone}?

The table below summarizes the classification performance of the two models on the Broward county data.\footnote{Following the ProPublica analysis, we restrict our attention to the $6150$ defendants in the data whose race was recorded as either African-American or Caucasian. }

\vspace{1em} 

\begin{center}
\begin{tabular}{lrrrrr}
    \toprule
    Model        & Accuracy & AUC  & PPV  & TNR  & TPR  \\
    \midrule                              
    \texttt{Priors}  & 0.64  & 0.67 & 0.62  & 0.71 & 0.56 \\
    \texttt{COMPAS}  & 0.66  & 0.70 & 0.65  & 0.75 & 0.55 \\

  \bottomrule
\end{tabular}
\end{center}

\vspace{0.25em}

\noindent The numeric scores were converted to classification rules using a cutoff of $2$ for \texttt{Priors} and $5$ for \texttt{COMPAS}.  These cutoffs were selected so that both models would classify approximately the same proportion of defendants as high-risk (42\% and 39\%, respectively).  While \texttt{COMPAS} is somewhat more accurate according to the various metrics, the difference in performance is overall not very large.  One might therefore be inclined to conclude that the choice of model does not make much difference, and that the results are similar perhaps because \texttt{COMPAS} likely puts a large weight on criminal history and thus reaches the same conclusion as \texttt{Priors}.  This conclusion is \emph{incorrect}.  As it turns out, the two classifiers disagree on 32\% of all cases.  Furthermore, as we will see in Section \ref{sec:compas_vs_priors}, they differ tremendously in terms of error rates and racial disparities for certain subgroups of defendants.  Model choice matters.

\vspace{0.5em} 

\noindent \textbf{Main contributions.}  We introduce a model comparison framework based on a recursive binary partitioning algorithm for automatically identifying subgroups in which the differences between two classification models are most pronounced.  The methods presented in this paper specifically focus on identifying subgroups where the models differ in terms of fairness-related quantities such as racial or gender disparities in error or acceptance rates.  Our methods can be applied to black-box models trained according to an unknown mechanism, do not require knowledge of what inputs the models use to make predictions, and do not require the models to use the same input variables.  

One noteworthy application of our method is in the model training phase, where one may wish to understand the effect that including a particular set of (potentially sensitive) variables has on the resulting classifications.  While there are certainly settings where using a sensitive attribute in decision-making is prohibited by law, this is far from always being the case.  Many domains permit the consideration of sensitive attributes when doing so improves the welfare of traditionally disadvantaged groups.  Indeed, depending on the problem setting, there may be good reason to expect predictive factors or mechanisms to differ across groups.  As \citet{hardt2014big} argues, ``statistical patterns that apply to the majority may be invalid within a minority group.''  In settings where using information on sensitive attributes may be permitted, it is important to understand the implications that this choice has for fairness.  Our framework provides a principled approach to investigating these kinds of issues.  We explore this matter further in the hypothetical lending example of Section \ref{sec:adult}.  


%

\vspace{-0.6em}

\subsection{Outline}

We begin with an overview of some related literature on model transparency and subgroup analysis.  In Section \ref{sec:method} we describe the general framework for our model comparison approach and provide some details on the implementation.  We conclude with experimental results where we investigate (i) how racial disparities differ across models in the ProPublica COMPAS data, and (ii) how gender disparities in acceptance rates change when additional sensitive attributes are added to a hypothetical model of creditworthiness.

\vspace{-0.6em}

\subsection{Related Work}

Within the algorithmic fairness literature, notable recent work has introduced new variable importance measures for quantifying the influence of variables on classification decisions (see, e.g., \citep{henelius2014peek, adler2016auditing, datta2015influence, datta2016algorithmic}).  A motivation common to much of this body work has been the problem of assessing whether sensitive attributes such as race or gender have direct or indirect influence on model outcomes.    We also note the recent work of \citet{zhang2016identifying},  which considers the single-model problem of identifying  subgroups in which the estimated event probabilities differ significantly from observed proportions.    This existing literature differs from our proposal in that we seek to quantify and characterize the difference in fairness across different models rather than to assess the direct or indirect influence of features in a single pre-trained model.  Our proposed method for characterizing differences in fairness across models has connections to recent work on subgroup analysis and recursive binary partitioning approaches for heterogeneous treatment effect estimation \citep{su2009subgroup, athey2015machine}.  
 
 \vspace{-0.6em}
 
\subsection{Fairness Metrics}

Throughout the paper we will make references to ``fairness metrics'' or ``disparities'', which often correspond to differences in a particular classification metric across groups.  For instance, \emph{statistical parity} or \emph{equal acceptance rates} with respect to a binary gender indicator would be satisfied if men and women were classified to the positive outcome at approximately equal rates.  \emph{False positive rate balance} with respect to race would be satisfied in the COMPAS example if non-reoffending Black defendants were misclassified as high-risk at the same rate as non-reoffending White defendants.  The work of \cite{hardt2016equality, kleinberg2016inherent, chouldechova2017fairlong, corbett2017algorithmic, berk2017fairness} describes numerous commonly used metrics, and provides a discussion of inherent trade-offs that exist between them.   \citet{romei2014multidisciplinary} provide a broader survey of multidisciplinary approaches to discrimination analysis that go beyond simple classification metrics.  

\section{Model Comparison framework}\label{sec:method}


We now describe our methodology for identifying subgroups in which a given
disparity differs across models.  The central components of this approach are as follows.  First, we define a quantity of interest, $\Delta$, that captures differences in model fairness, and we show how this quantity is a simple function of the parameters of an exponential family model.  We then apply a recursive binary partitioning algorithm that uses a score-type test for $\Delta$ to partition the covariate space into regions within which $\Delta$ is homogeneous.
%

\vspace{0.25em}

\noindent \textbf{Notation.}  We begin with some notation.
Let $A\in\{a_1,a_2\}$ indicate a sensitive binary attribute (e.g., race in the
COMPAS example), and let
$Y \in \{0,1\}$ indicate the true outcome (e.g., 2-year recidivism).  Let
$\hat{Y}_{m_1}, \hat{Y}_{m_2}\in\{0,1\}$ denote the classifications made by two
classifiers, $m_1$ and $m_2$.  Due to space limitations, we focus our
description on disparities in the False Positive Rate.  Extensions to other
fairness metrics involving expressions of the form $\hat Y \mid A, Y$ (e.g., FNR, acceptance rates) are entirely analogous, and are discussed in Section \ref{sec:extensions}.  

The false
positive rate (FPR) for classifier $m_j$, among individuals in group $A=a$ is
denoted
$\mathrm{FPR}_{m_j}^a = \P(\hat{Y}_{m_j}=1|Y=0, A=a)$.  Disparities in FPR across
values of $A$ for model $m_j$ are captured by $\mathrm{FPR}_{m_j}^{a_2} -
\mathrm{FPR}_{m_j}^{a_1}$, and hence differences in these disparities between
classifiers can be captured by the difference-in-differences of the FPR:
\begin{align}
  \Delta &= \left(\mathrm{FPR}_{m_2}^{a_2} - \mathrm{FPR}_{m_2}^{a_1}\right) -
  \left(\mathrm{FPR}_{m_1}^{a_2} - \mathrm{FPR}_{m_1}^{a_1}\right).
  \label{eq:delta}
\end{align}
We focus on the difference-in-differences instead of difference in \emph{absolute} differences because it is important to be able to capture cases where the disparity differs in sign between two models but not necessarily in magnitude.  

The goal of the proposed method is to partition the covariate space into
subgroups such that $\Delta$ is homogeneous within each subgroup and different
between subgroups.  We say that $\Delta$ is homogeneous within a group if
that group cannot be partitioned into subgroups with significantly different
$\Delta$ values.  In Section \ref{sec:partition}, we will describe the
application of test-based recursive partitioning to obtain subgroups
homogeneous in $\Delta$.  In Section \ref{sec:model}, we describe the
likelihood model which underlies the tests of homogeneity.


\begin{table*}[ht]
  \begin{minipage}{0.49\textwidth}
    \centering
  \begin{tabular}{r|cc}
    \toprule
    $A = a_1$ & $\hat{Y}_{m_2}=0$ & $\hat{Y}_{m_2}=1$\\
    \hline
    $\hat{Y}_{m_1}=0$ & $p_{00}^{a_1}$ & $p_{01}^{a_1}$\\
    $\hat{Y}_{m_1}=1$ & $p_{10}^{a_1}$ & $p_{11}^{a_1}$ \\
    \bottomrule
  \end{tabular}
  \end{minipage}
  \begin{minipage}{0.49\textwidth}
    \centering
\begin{tabular}{r|cc}
  \toprule
    $A = a_2$ & $\hat{Y}_{m_2}=0$ & $\hat{Y}_{m_2}=1$\\
    \hline
    $\hat{Y}_{m_1}=0$ & $p_{00}^{a_2}$ & $p_{01}^{a_2}$\\
    $\hat{Y}_{m_1}=1$ & $p_{10}^{a_2}$ & $p_{11}^{a_2}$\\
    \bottomrule
  \end{tabular}
  \end{minipage}
  \caption{Conditional on the sensitive attribute $A\in\{a_1,a_2\}$, the observed  $(\hat{Y}_{m_1},\hat{Y}_{m_2})$ can be modeled as a multinomial 
  with parameters $p_{ij}^a$.  For the case of FPR disparities discussed in Section \ref{sec:model}, all quantities in this table should be interpreted as being further conditioned on $Y = 0$. }
  \label{tab:multinom}
\end{table*}

\subsection{Partitioning scheme}
\label{sec:partition}

Given a set of partitioning covariates $X_1, \ldots, X_p$---which need not
correspond in any way to the inputs used by either classifier---we recursively
partition the covariate space using tests of the homogeneity of $\Delta$.  Our
partitioning procedure follows the approach of \cite{rpartykit, rmob} for
model-based recursive binary partitioning, and relies on a modified version of the
corresponding software.  Our implementation uses custom fitting functions supplied to  the \texttt{R} package \texttt{partykit}.   We briefly describe the procedure for the simple
case where all of the splitting variables are categorical.  The approach and software both fully extend to also handle numeric and ordinal variables.

Let $K_j$ denote the number of distinct levels of variable $X_j$, and let $\Delta^j_k$ denote the (population) value of $\Delta$ in level $k$ of variable $X_j$.   Beginning with all observations in the root node,
recursively split according to the following procedure:

\begin{enumerate}

  \item For each partitioning variable $j = 1, \ldots, p$, apply a score-type test (see Section \ref{sec:model} and Appendix \ref{appendix:likelihood}) for detecting when $\Delta^j_k$ varies across the levels $k \in \{1, \ldots, K_j\}$.  \\
  Select the splitting variable with the most
    significant difference in $\Delta$ (the smallest $p$-value for this test).  

  \item For the selected variable, partition its levels into the two groups
    which minimize the total deviance of the resulting model.

\end{enumerate}

The recursion terminates when nodes cannot be further split without falling below a user-specified minimum size
threshold, or no further splits can be identified for which the Bonferroni-adjusted $p$-value is smaller than a user-specified significance threshold.  As a final step,
the tree is pruned to eliminate splits where the differences in $\Delta$ are not of practical significance, but which were statistically significant due to large sample sizes. 

This partitioning scheme produces what we will refer to as a \emph{parameter instability tree}, with splits defined based on individual
covariates, similar to the familiar trees produced by CART
\citep{breiman1984classification} in classification settings.  The leaf nodes
of the tree correspond to subgroups where $\Delta$ appears homogeneous.
Figure \ref{fig:compas_tree} shows an example of the parameter instability tree for the COMPAS data, with $\Delta$ taken to be the difference in racial FPR disparity between the \texttt{Priors} and \texttt{COMPAS} model.  Section \ref{sec:compas_vs_priors} provides more details on the experimental setup.



%

\subsection{Modeling classifications}
\label{sec:model}

To carry out the recursive partitioning of Section \ref{sec:partition}, we need
a model for the classifications which is (a) reasonable, and (b) easily
captures $\Delta$ in its parametrization.  Given a population, a natural
joint model of classifications
$(\hat{Y}_{m_1}, \hat{Y}_{m_2})$ is as a multinomial conditional on sensitive
attribute $A \in \{a_1,a_2\}$, as illustrated in Table \ref{tab:multinom}. This multinomial is parameterized by the probabilities
\begin{align*}
  p_{ij}^a &= \P(\hat{Y}_{m_1}=i, \hat{Y}_{m_2}=j| Y=0, A=a),
\end{align*}
where $a\in\{a_1,a_2\}$.  The conditional multinomial is a convenient
formulation, since all relevant FPR quantities can be represented in terms of
these parameters:
\begin{align*}
  \mathrm{FPR}_{m_2}^{a_2} &= p_{01}^{a_2} + p_{11}^{a_2}\qquad
  \mathrm{FPR}_{m_2}^{a_1} = p_{01}^{a_1} + p_{11}^{a_1}\\
  \mathrm{FPR}_{m_1}^{a_2} &= p_{10}^{a_2} + p_{11}^{a_2}\qquad
  \mathrm{FPR}_{m_1}^{a_1} = p_{10}^{a_1} + p_{11}^{a_1}
\end{align*}

An important observation is that, for the purpose of computing the quantity of interest $\Delta$, it suffices to consider the coarser conditional multinomial over the three events
$\{\hat{Y}_{m_1}=0,\hat{Y}_{m_2}=1\}, \{\hat{Y}_{m_1}=1,\hat{Y}_{m_2}=0\},
\{\hat{Y}_{m_1}=\hat{Y}_{m_2}\}$. This reduced multinomial is parameterized by
$p_{10}^{a_1}, p_{01}^{a_1}, p_{10}^{a_2}, p_{01}^{a_2}$.  We summarize this observation in the proposition below.  

\begin{proposition}
  The FPR difference-in-difference, $\Delta$, can be written as
 \begin{align*}
  \Delta &= 
  p_{01}^{a_2} - p_{01}^{a_1} - p_{10}^{a_2} +p_{10}^{a_1} 
\end{align*}
\end{proposition}
\begin{proof}
  The proposition follows directly from the definition of $\Delta$ and the
  identity,
\begin{align*}
  &\P(\hat{Y}_{m_2}=1|Y=0,A=a) - \P(\hat{Y}_{m_1}=1|Y=0,A=a)\\
  &= p_{01}^a - p_{10}^a.
\end{align*}
\vspace{-1em}
\end{proof}

To obtain the score-type test statistic for $\Delta$ required in Step (1) of the partitioning scheme described in Section \ref{sec:partition}, we further reparameterize the model according to the transformations:
\newcommand{\basepw}{\eta^+}
\newcommand{\basemw}{\eta^-}
\newcommand{\diffpb}{\delta}
\newcommand{\diffmb}{\Delta}
\begin{align*}
  \basepw &= p_{01}^{a_1} + p_{10}^{a_1}\qquad  \diffpb = p_{01}^{a_2} + p_{10}^{a_2} - \basepw\\
  \basemw &= p_{01}^{a_1} - p_{10}^{a_1}\qquad  \Delta = p_{01}^{a_2} - p_{10}^{a_2} - \basemw .
\end{align*}
In this parameterization, $\basepw, \diffpb$ and $\basemw$ are treated as nuisance parameters in the model likelihood for the purpose of forming the score-type test statistic for $\Delta$.  A more complete derivation of the test statistic can be found in Appendix \ref{appendix:likelihood}.

 \subsection{Extensions}\label{sec:extensions}

 The methodology in Section \ref{sec:method} focuses on identifying subgroups
 with where two models differ in terms their FPR disparities.  To target disparities in False Negative Rates (FNR), the same procedure can be
 carried out by conditioning on $Y=1$ instead of $Y=0$, and exchanging the role of
 $\hat{Y}_m=1$ and $\hat{Y}_m=0$ in expression \eqref{eq:delta}.   $\Delta$ would then correspond precisely to the difference in FNR disparity. Similarly, to target disparities in the acceptance rates $\P(\hat{Y}_m = 1| A=a)$,
 the procedure can be carried out without conditioning
 on $Y$.   Certain other metrics may similarly be considered.  

 In principle, the methodology can also be extended to sensitive attributes $A$ that have more than two levels.  One would first need to define a quantity $\Delta$ that reflects the disparity of model predictions with respect to $A$.  For instance, in the case of acceptance rates, $\Delta$ could be taken to be the variance in acceptance rates across race (now understood to be non-binary).  An extension of the proposed procedure to this quantity would thus identify subgroups where one model exhibits greater variability in acceptance rates across race compared to another model.   Alternatively, the proposed approach can be applied directly in an all-pairs or one-versus-all manner.
 
Lastly, we note that the score-type test for testing the null hypothesis in Step (1) of Section \ref{sec:partition} can be replaced with any other valid statistical test.  One could thus use a test that has greater power against particular types of alternatives.  Note, however, that the score test is a computationally efficient choice.   This is because, unlike most tests, the score test only requires that maximum likelihood parameters be computed under the null.  This obviates the need for model refitting under the alternative for each splitting variable.

%

\begin{figure}[t]
\includegraphics[width = 0.51\textwidth]{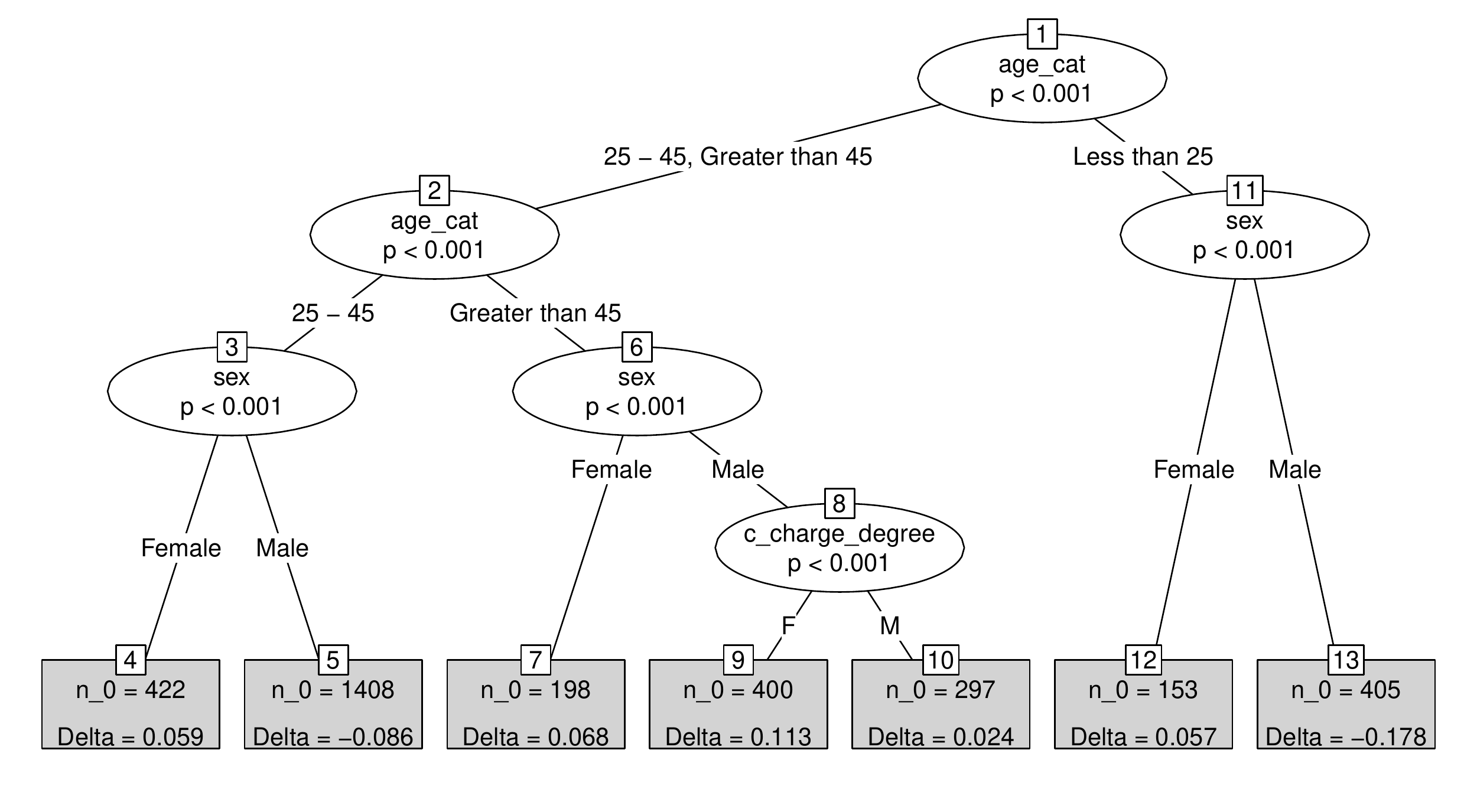}
\caption{Parameter instability tree for identifying differences in racial FPR disparities between \texttt{COMPAS} and \texttt{Priors}.  Sample sizes $n_0$ indicate number of observations in the terminal node for which $Y = 0$.  Negative values of $\Delta$ correspond to subgroups where the FPR disparity in favor of White defendants is smaller for \texttt{Priors} than for \texttt{COMPAS}.}
\label{fig:compas_tree}
\end{figure}

\begin{figure}[t]
\includegraphics[width = 0.51\textwidth]{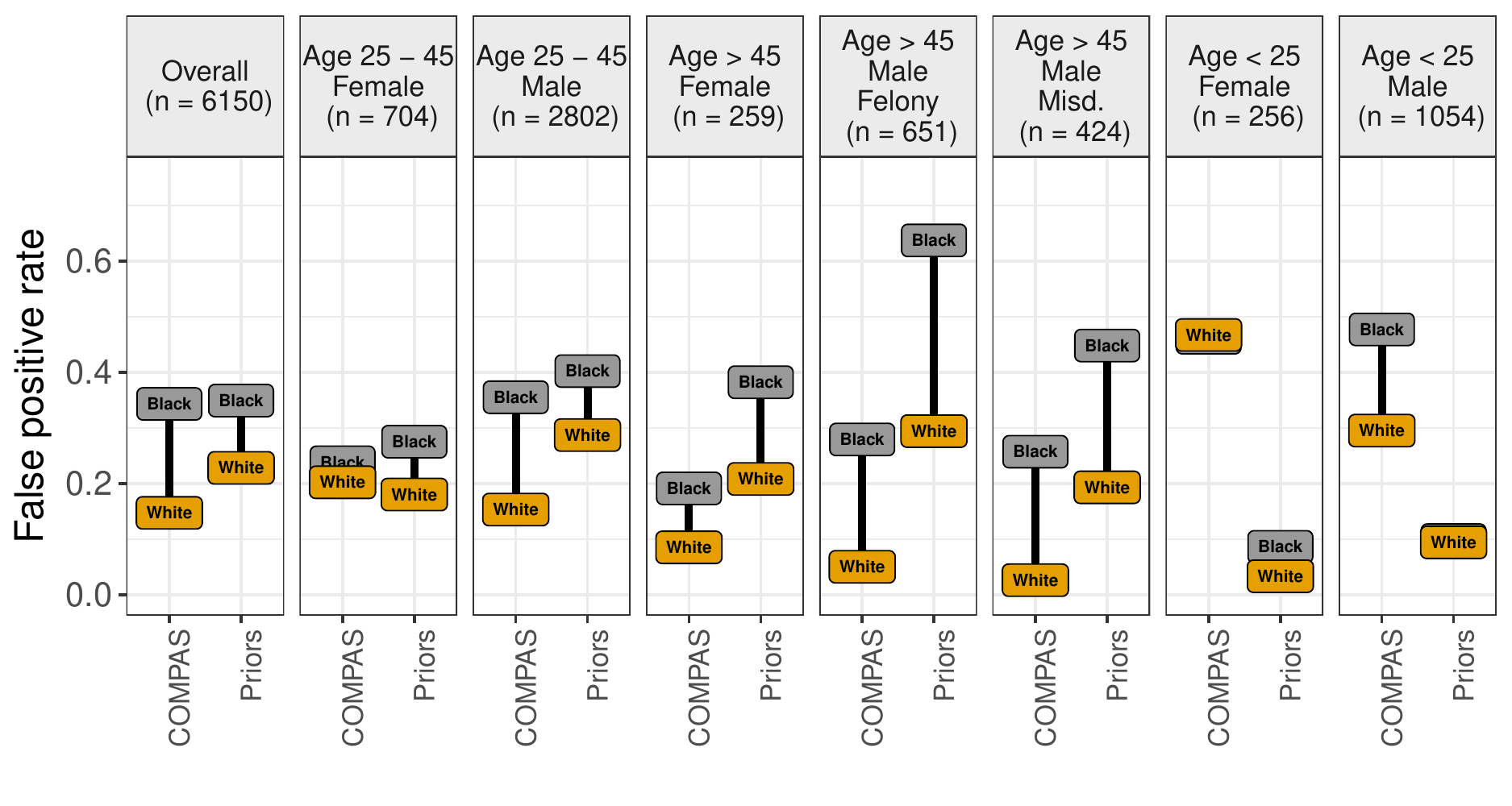}
\caption{Differences in racial FPR disparities between \texttt{COMPAS} and \texttt{Priors}.}
\label{fig:compas_fpr}
\end{figure}

\section{Evaluation}


\subsection{Recidivism risk prediction} \label{sec:compas_vs_priors}


We begin by revisiting our motivating example with ProPublica's COMPAS data from Broward County, Florida. So far we have seen that the \texttt{COMPAS} score performs similarly to the priors count, \texttt{Priors}, in terms of overall classification metrics.  To delve deeper into differences between these two recidivism prediction models, we apply our method to identify subgroups where \texttt{COMPAS} and \texttt{Priors} differ in terms of the \emph{disparity in false positive rates between Black and White defendants}.  The candidate splitting variables are taken to be \verb|sex|, \verb|age_cat|, \verb|c_charge_degree|, \verb|juv_misd_count|, \\ \verb|juv_fel_count|, and \verb|juv_other_count|.  Figure~\ref{fig:compas_tree} shows the resulting parameter instability tree, and Figure~\ref{fig:compas_fpr} provides a more easily interpretable representation of the findings.  

Our method identifies $7$ subgroups defined in terms of \verb|sex|, \verb|age_cat| and \verb|c_charge_degree| splits where the extent or nature of the disparity in FPR between Black and White defendants is different between the two models.  For instance, as we can clearly see in rightmost panel of Figure~\ref{fig:compas_fpr}, the racial FPR disparity among young men is large for \texttt{COMPAS} but is nearly $0$ for \texttt{Priors}.  

We emphasize two key points.  First, we observe that the Overall difference in racial FPR disparity is not reflective of differences at the subgroup level.  Furthermore, we note that while the differences in FPR across the $7$ subgroups are at least in part due to differences in recidivism prevalence across the subgroups, the same argument does not explain the differences between \texttt{COMPAS} and \texttt{Priors} \emph{within} the subgroups. 

\subsection{Sensitive attributes as inputs} \label{sec:adult}

\begin{figure}[t]
\includegraphics[width = 0.51\textwidth]{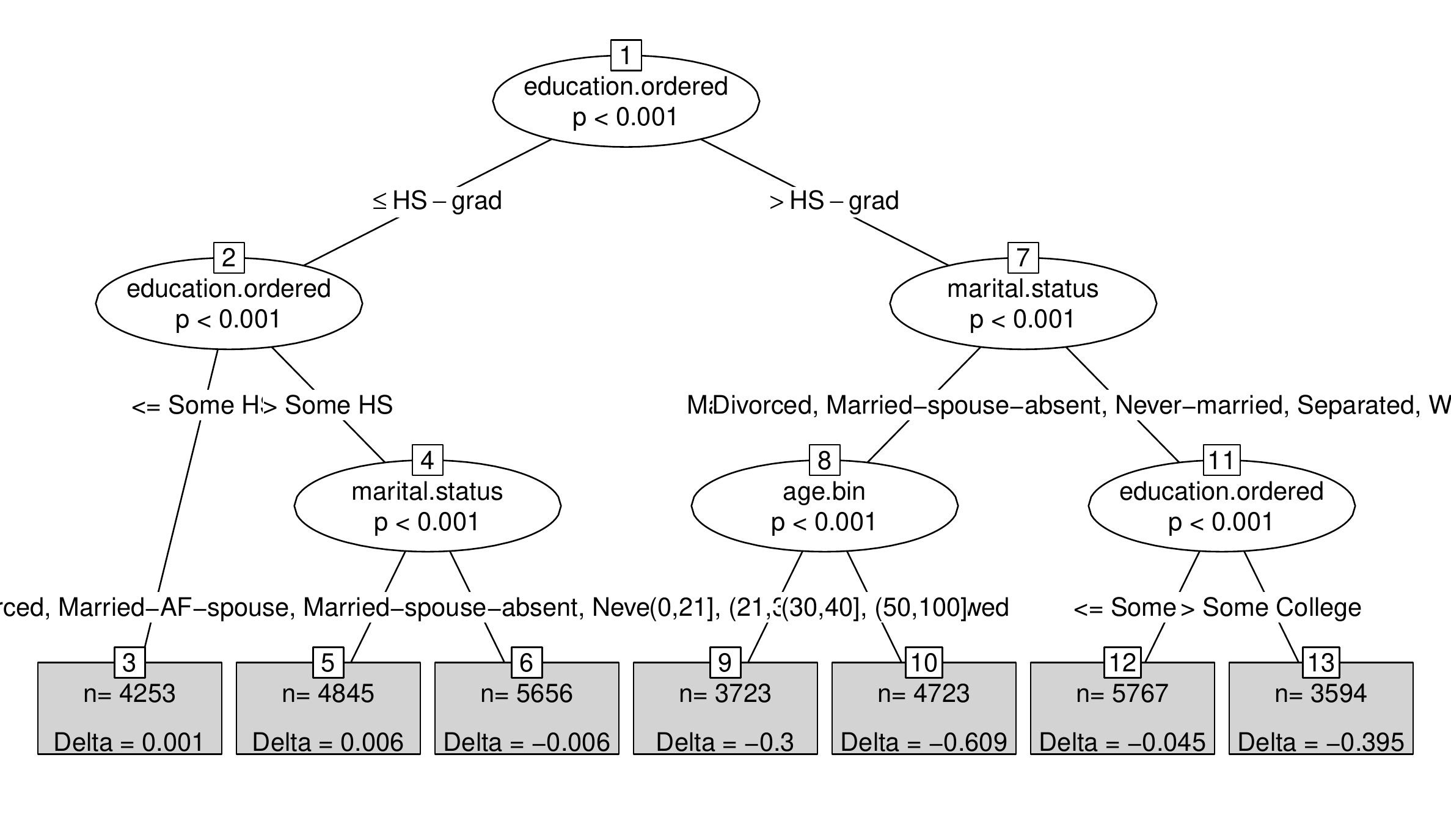}
\caption{Parameter instability tree for identifying differences in acceptance (``lending'') rates between men and women between the \texttt{Small} and \texttt{Full} models.  Negative values of $\Delta$ correspond to subgroups where the acceptance rate disparity in favor of Male applicants is smaller for \texttt{Full} than for \texttt{Small}. Nodes \{3, 5, 6\} are later collapsed in the pruning stage as the corresponding $\Delta$'s of $0.001, 0.006$ and $-0.006$, while statistically significant due to large sample sizes, are not of practical significance.}
\label{fig:adult_tree}
\end{figure}

For our next example we use the \texttt{Adult} data set from the UCI database \cite{Lichman:2013} to frame a hypothetical lending problem.  We fit two random forest models to the data to predict whether individuals are in the \verb|>50K| income (``loan-worthy'') category.  The \texttt{Small} model uses \verb|sex, age, workclass, education.years| as inputs, while the \texttt{Full} model \emph{additionally} uses \verb|race| and \verb|marital.status|, both of which are typically considered to be sensitive attributes.  While we do not claim that either model is realistic, this setup does illustrate an interesting phenomenon.  

We apply our method to identify subgroups where the \emph{disparity in lending rates between Male and Female applicants differs between the \texttt{Small} and \texttt{Full} model}.  More precisely, $\Delta$ in this example is taken to be:
\begin{align*}
  \Delta &= \left(\P(\hat Y_\texttt{Full} = 1 \mid \texttt{Male}) - \P(\hat Y_\texttt{Full} = 1 \mid \texttt{Female})\right) - \\
  &\left(\P(\hat Y_\texttt{Small} = 1 \mid \texttt{Male}) - \P(\hat Y_\texttt{Small} = 1 \mid \texttt{Female})\right).
\end{align*}
  The candidate splitting variables are taken to be \verb|education|, \verb|age|, \verb|marital.status| and \verb|race|.    Figure~\ref{fig:adult_tree} shows the resulting parameter instability tree, and Figure~\ref{fig:adult_accept} provides a more interpretable representation of the results.  Unlike in the COMPAS example, the number of terminal nodes presented in the tree differs from the number of subgroups presented in the Figure~\ref{fig:adult_accept} summary.  This is because the tree is shown \emph{prior} to pruning, a final step that collapses nodes 3,5 and 6 into a single \verb|{Education <= High School}| subgroup.

We observe that overall acceptance (lending) rates go up for both men and women when marital status and race is included in the model.  We also find that the gender disparity in lending rates \emph{decreases}---and even \emph{inverts}---among Married individuals who have more than a High School education.  The disparity also decreases considerably among unmarried individuals with at least a College education.  However, this is largely due to the massive drop in lending rates among Men in this subgroup.




{
\begin{figure}[t]
\includegraphics[width = 0.51\textwidth]{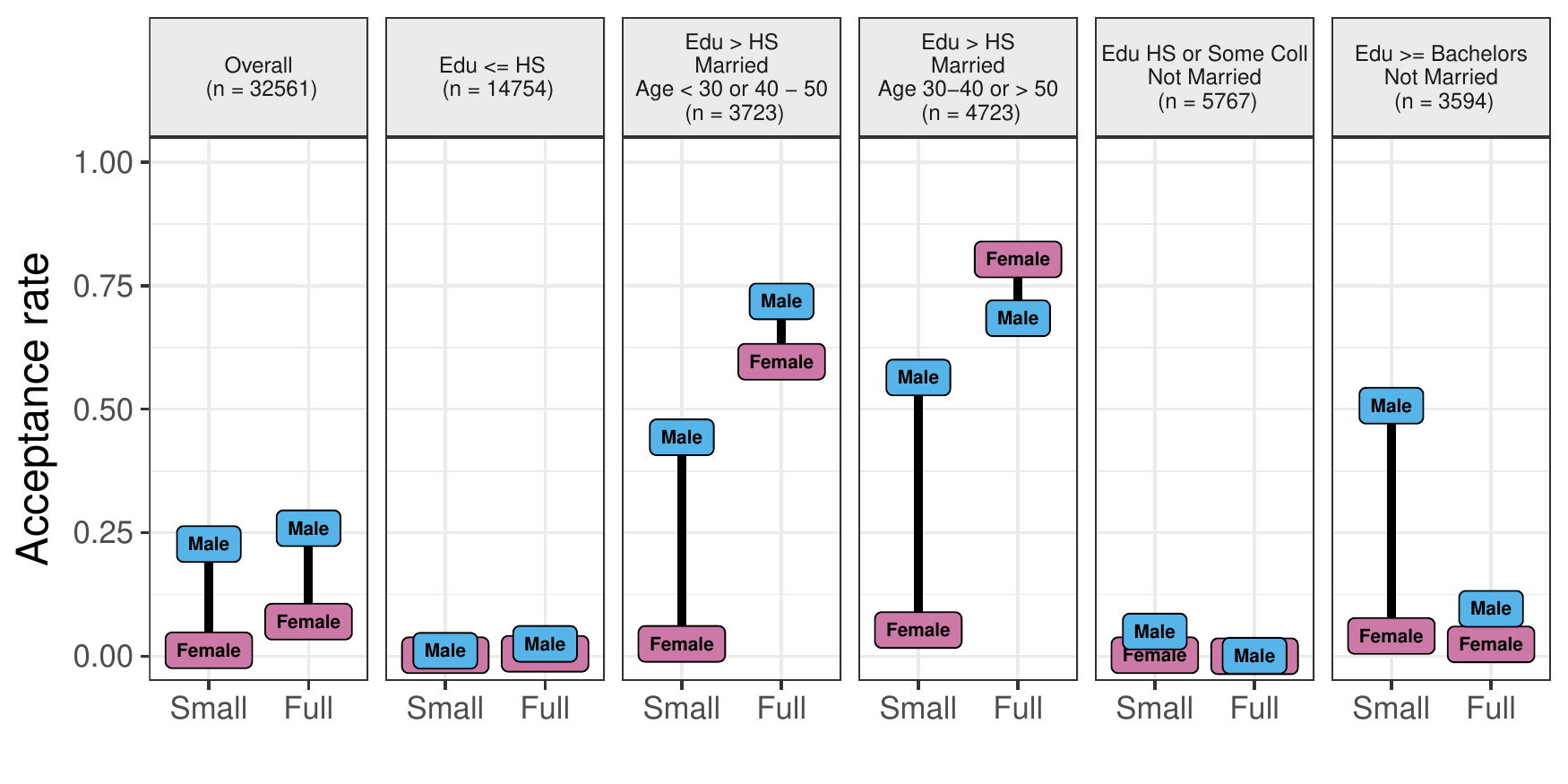}
\caption{Differences in acceptance (``lending'') rates between men and women when \texttt{marital.status} and \texttt{race} are \emph{not} included as inputs (\texttt{Small}) compared to when they are (\texttt{Full}).}
\label{fig:adult_accept}
\end{figure}
}

\section{Conclusion}

This paper introduced a test-based recursive binary partitioning approach to identifying subgroups where two models differ considerably in terms of their fairness properties.  Using examples in recidivism prediction and lending, we showed how this approach can be used to detect large subgroup differences in fairness that are not apparent from an overall performance comparison.  The methodology can be further extended to target other kinds of disparity parameters and to use other statistical tests for parameter instability.  

\section{Acknowledgements}  
We thank the anonymous FAT/ML referees for their helpful comments on the initial version of this manuscript.


%
%


\bibliographystyle{ACM-Reference-Format}
\bibliography{modelcomp} 


\appendix

\section{Score test derivation} \label{appendix:likelihood}

In this section we present the derivation of the likelihood and score-type test used in the partitioning scheme described in Section \ref{sec:partition}.  This derivation is carried out for the FPR difference-in-differences parameter as defined in expression \eqref{eq:delta}.  Tests for other parameters of interest may be derived analogously.

\subsection*{Parametrization in terms of $\Delta$}

For identifying partitions of covariate space on which $\Delta$ is homogeneous, it is helpful to write the
distribution and its likelihood directly in terms of $\Delta$.  We use the
following reparameterization of the multinomial
\begin{align*}
  \basepw &= p_{01}^{a_1} + p_{10}^{a_1}  &\diffpb = p_{01}^{a_2} + p_{10}^{a_2} - \basepw\\
  \basemw &= p_{01}^{a_1} - p_{10}^{a_1}  &\diffmb = p_{01}^{a_2} - p_{10}^{a_2} - \basemw .
\end{align*}
The parameterization $(\basepw, \basemw, \diffpb, \diffmb)$ is equivalent to \\
$(p_{01}^{a_1}, p_{10}^{a_1}, p_{01}^{a_1}, p_{10}^{a_1})$, since
\begin{align*}
  p_{01}^{a_1} &= \left(\basepw + \basemw\right)/2, \quad
  p_{01}^{a_2} = \left(\basepw + \diffpb + \basemw + \diffmb\right)/2\\
  p_{10}^{a_1} &= \left(\basepw - \basemw\right)/2, \quad
  p_{10}^{a_2} = \left(\basepw + \diffpb - \basemw - \diffmb\right)/2.
\end{align*}

\vspace{2em}

\subsection*{Likelihood and score}

The log-likelihood for this multinomial model of a \emph{sample} of size $n$
with parameter $\theta = (p_{00}, p_{01}, p_{10}, p_{11})$ is 
 \begin{align*}
   \ell(\theta)   &= \sum_{a\in\{a_1,a_2\}} \left(n_{01}^a \log p_{01}^a +
   n_{10}^a\log p_{10}^a + n_\bullet^a\log p_\bullet^a\right),
 \end{align*}
 where $n_{ij}^a$ is the number of observations in group $A=a$
 classified to $\hat{Y}_{m_1} = i$ and $\hat{Y}_{m_2}=j$.  Additionally, 
 $n_\bullet^a = n_{00}^a + n_{11}^a$, and $p_\bullet^a = 1-p_{10}^a - p_{01}^a$.

 For constructing the test of homogeneity, it is useful to have an expression
 for the score function for a \emph{single observation} with respect to the parameters
$(\basepw, \basemw, \diffpb, \Delta)$.  This is given by:
\begin{align*}
  \dot\ell_i(\theta) &= \left(
   \frac{\partial\ell_i(\theta)}{\partial \basepw},
   \frac{\partial\ell_i(\theta)}{\partial \basemw},
   \frac{\partial\ell_i(\theta)}{\partial \diffpb},
   \frac{\partial \ell_i(\theta)}{\partial \Delta} 
   \right)^T ,
 \end{align*}
 where
 \begin{align*}
   \frac{\partial\ell_i(\theta)}{\partial \basepw} &=
   \frac{\Ind_{01}^{a_1}}{2p_{01}^{a_1}} + 
   \frac{\Ind_{10}^{a_1}}{2p_{10}^{a_1}} + 
   \frac{\Ind_{01}^{a_2}}{2p_{01}^{a_2}} + 
   \frac{\Ind_{10}^{a_2}}{2p_{10}^{a_2}} - 
   \frac{\Ind_\bullet^{a_1}}{p_\bullet^{a_1}}-
   \frac{\Ind_\bullet^{a_2}}{p_\bullet^{a_2}}\\
 \frac{\partial\ell_i(\theta)}{\partial \basemw} &=
   \frac{\Ind_{01}^{a_1}}{2p_{01}^{a_1}} - 
   \frac{\Ind_{10}^{a_1}}{2p_{10}^{a_1}} + 
   \frac{\Ind_{01}^{a_2}}{2p_{01}^{a_2}} - 
   \frac{\Ind_{10}^{a_2}}{2p_{10}^{a_2}} \\ 
   \frac{\partial\ell_i(\theta)}{\partial \diffpb} &=
   \frac{\Ind_{01}^{a_2}}{2p_{01}^{a_2}} + 
   \frac{\Ind_{10}^{a_2}}{2p_{10}^{a_2}} - 
   \frac{\Ind_\bullet^{a_2}}{p_\bullet^{a_2}}\\
   \frac{\partial \ell_i(\theta)}{\partial \Delta} &= 
   \frac{\Ind_{01}^{a_2}}{2p_{01}^{a_2}} -
   \frac{\Ind_{10}^{a_2}}{2p_{10}^{a_2}},
 \end{align*}
 with $\Ind_{ij}^a = \begin{cases} 1 & \text{if } \hat{Y}_{m_1}=i,
   \hat{Y}_{m_2}=j, \text{ and } A=a\\
 0 & \text{otherwise}\end{cases}$ , \\
 and $\Ind_{\bullet}^a =
 \Ind_{00}^a+\Ind_{11}^a$.
 In this notation, the score for the full sample is $\dot\ell(\theta) = \sum_{i=1}^n
 \dot\ell_i(\theta)$.

\newpage
\subsection*{Test statistic}

To test for homogeneity, we use the score-based Lagrange Multiplier test
statistic of \cite{merkle2013tests}.  Consider a categorical splitting variable
$X \in \R^n$ with levels in $\{1,\dots,K\}$.  Let $\hat \theta$ denote the MLE of $\theta$ under the null hypothesis that $\theta$ is constant across every level of $X$.  While the null is equality of the entire parameter vector $\theta$, our test statistic is constructed to have power to detect differences in $\Delta$.  The test statistic is constructed as follows:
\begin{align*}
  T &= \sum_{k=1}^K \frac{1}{\#\{X_i = k\}} \left(\sum_{i: X_i=k}
  \dot\ell_i(\hat\theta)^T\hat{I}(\hat\theta)^{-1/2}e_4\right)^2 ,
\end{align*}
where $\hat{I}(\theta) = \frac{1}{n}\sum_{i=1}^n
\dot\ell_i(\hat\theta)\dot\ell_i(\hat\theta)^T$ and $e_4 = (0,0,0,1)^T$ ($\Delta$
corresponds to the fourth parameter).  Under the null hypothesis that $\theta$ is constant across every level of $X$, $T$ asymptotically follows the $\chi_{K-1}^2$
distribution \cite{hjort2002tests}.
\end{document}